\documentclass{eptcs}
\providecommand{\event}{QPL 2016}
\title{Ambiguity and Incomplete Information in Categorical Models of Language}
\author{Dan Marsden
  \institute{University of Oxford}
  \email{daniel.marsden@cs.ox.ac.uk}}
\newcommand{\authorrunning}{D. Marsden}
\newcommand{\titlerunning}{Mixing and Ambiguity in Categorical Models of Language}

\usepackage{color}
\usepackage{amsmath, amsthm, amssymb, amsfonts}
\usepackage{xspace}
\usepackage{tikz}
\usetikzlibrary{calc}

\theoremstyle{plain}
\newtheorem{theorem}{Theorem}
\newtheorem{lemma}[theorem]{Lemma}

\newtheorem{proposition}[theorem]{Proposition}
\theoremstyle{definition}
\newtheorem{definition}[theorem]{Definition}
\newtheorem{example}[theorem]{Example}
\newtheorem{counter}[theorem]{Counterexample}
\theoremstyle{remark}
\newtheorem{remark}[theorem]{Remark}

\numberwithin{theorem}{section}

\newcommand{\eilmo}[1]{\ensuremath{\operatorname{\operatorname{EM}}(#1)}}
\newcommand{\ket}[1]{\ensuremath{| #1 \rangle}}
\newcommand{\catname}[1]{\mathbf{#1}}
\newcommand{\convex}{\ensuremath{\catname{Convex}}\xspace}
\newcommand{\subconvex}{\ensuremath{\catname{Subconvex}}\xspace}

\newcommand{\standardrel}{\ensuremath{\catname{Rel}}\xspace}

\newcommand{\cpm}[1]{\ensuremath{\catname{CPM}(#1)}\xspace}
\newcommand{\fdhilb}{\ensuremath{\catname{FdHilb}}\xspace}

\newcommand{\cset}{\ensuremath{\catname{Set}}\xspace}
\newcommand{\pset}{\ensuremath{\catname{Set_\bullet}}\xspace}

\newcommand{\jslat}{\ensuremath{\catname{JSLat}}\xspace}
\newcommand{\ajslat}{\ensuremath{\catname{AJSLat}}\xspace}
\newcommand{\subdistmonad}{S}
\newcommand{\distmonad}{D}
\newcommand{\powersetmonad}{P}
\newcommand{\finitepowersetmonad}{\powersetmonad_\omega}
\newcommand{\nonemptyfinitepowersetmonad}{\finitepowersetmonad^+}
\newcommand{\define}[1]{{\bf #1}}
\DeclareMathOperator{\support}{supp}

\DeclareMathOperator{\strength}{st}
\DeclareMathOperator{\costrength}{st'}
\DeclareMathOperator{\doublestrength}{dst}

\begin{document}
\maketitle

\begin{abstract}
  We investigate notions of ambiguity and partial information in
  categorical distributional models of natural language.
  Probabilistic ambiguity has previously been studied in
  \cite{PiedeleuKartsaklisCoeckeSadrzadeh2015, Piedeleu2014, Kartsaklis2014} using Selinger's
  CPM construction. This construction works well for models built upon vector spaces, as has
  been shown in quantum computational applications. Unfortunately, it
  doesn't seem to provide a satisfactory method for introducing mixing in other compact closed
  categories such as the category of sets and binary relations. We therefore lack a uniform strategy for
  extending a category to model imprecise linguistic information.

  In this work we adopt a different approach. We analyze different forms of ambiguous
  and incomplete information, both with and without quantitative probabilistic data. Each
  scheme then corresponds to a suitable enrichment of the category in which we model language.
  We view different monads as encapsulating the informational behaviour of interest, by
  analogy with their use in modelling side effects in computation. Previous results of Jacobs
  then allow us to systematically construct suitable bases for enrichment.
  
  We show that we can freely enrich arbitrary dagger compact closed categories in order
  to capture all the phenomena of interest,
  whilst retaining the important dagger compact closed structure. This allows us to construct
  a model with real convex combination of binary relations that makes non-trivial use of the
  scalars. Finally we relate our various different enrichments, showing that finite subconvex algebra
  enrichment covers all the effects under consideration.
\end{abstract}

\section{Introduction}
The categorical distributional approach to natural language processing~\cite{CoeckeSadrzadehClark2010} aims
to construct the meaning of a sentence from its grammatical structure and
the meanings of its parts. The grammatical structure of language can be
described using pregroup grammars~\cite{Lambek1997}. The distributional approach to natural language
models the meanings of word as vectors of statistics in finite dimensional real
vector spaces. Both the category of real vector spaces and linear maps, and pregroups are examples
of monoidal categories in which objects have duals.
This common structure is the key observation that allows us to functorially transfer grammatical structure to linear maps modelling meaning in a compositional manner.

Recent work in categorical models of language has investigated ambiguity, for example words with multiple meanings,
\cite{PiedeleuKartsaklisCoeckeSadrzadeh2015, Piedeleu2014, Kartsaklis2014}.
These papers exploit analogies with quantum mechanics, using density matrices to model
partial and ambiguous information.
In order to do this, the category in which meanings are interpreted must change.
Selinger's CPM construction~\cite{Selinger2007} is exploited to construct compact closed categories in which ambiguity can be described.
This construction takes a compact closed category~$\mathcal{C}$ and produces a new compact closed category~$\cpm{\mathcal{C}}$.
When applied to the category of finite dimensional Hilbert spaces and linear maps, \fdhilb, the resulting category~\cpm{\fdhilb}
is equivalent to the category of finite dimensional Hilbert spaces and completely positive maps. In this way the construction takes the setting for
pure state quantum mechanics and produces exactly the right setting for mixed state quantum mechanics. It is therefore tempting
to consider the CPM construction as a mixing device for compact closed categories.
This perspective was adopted in~\cite{PiedeleuKartsaklisCoeckeSadrzadeh2015} by starting with the category \standardrel of sets and binary relations
and interpreting \cpm{\standardrel} as a toy model of ambiguity.
As argued in \cite{Marsden2015, Gogioso2015}, there are aspects of~\cpm{\standardrel} that conflict with its interpretation as a setting for mixing. Specifically:
\begin{itemize}
\item There are pure states that can be formed as a convex mixture of two mixed states
\item There are convex combinations of distinct pure states that give pure states
\end{itemize}
It could be argued that this anomalous behaviour is due to the restricted nature of the scalars in \standardrel. This then points to another weakness of
\cpm{\standardrel} as a setting for ambiguity and mixing. We may wish to say that the word bank is $90\%$ likely to mean a financial institution
and $10\%$ likely to refer to the boundary of a river. We can express this in \cpm{\fdhilb} as the scalars are sufficiently rich. In~\cpm{\standardrel} the
best we can do is say that it's either one or the other, with all the quantitative information being lost.

In this paper we investigate some alternative models of ambiguity in the compact closed setting. We consider
a variety of different interpretations of what it might mean to have ambiguous or limited information, and
how these can be described mathematically. Specifically:
\begin{itemize}
\item In section~\ref{sec:unquantified} we describe constructions that model ambiguous and incomplete information
  in a non-quantitative manner. We show that compact closed categories can be freely extended so as to allow the
  modelling of incomplete information, ambiguity and a mixture of both phenomena.
\item In section~\ref{sec:quantified} we extend our constructions to describe quantitative ambiguous and incomplete
  information. Again we show that compact closed categories can be freely extended in order to model such features.
\item Proposition~\ref{prop:embeddings} shows that these various informational notions embed into each other.
\end{itemize}
Our general perspective is to enrich the homsets of our categories in order to model the language features we are interested in.
In order to do this systematically, we exploit some basic monad theory. Monads are commonly used to describe computational
effects such as non-determinism, exceptions and continuations~\cite{Moggi1991, Wadler1995}.
In our case we instead view them as models of informational effects in natural language applications.
Monads have previously been used in models of natural language, see for example~\cite{Shan2001}.
Work of Jacobs~\cite{Jacobs1994} provides the connection between a certain class of monads and categories
that provide good bases for enriched category theory.
We aim to give explicit constructions of the categories of concrete interest throughout, rather than pursuing a policy of maximum abstraction.
Our systematic approach means we should be able to incorporate additional informational features in a similar manner.

We assume familiarity with elementary category theory, and the notions of compact closed and dagger compact closed categories \cite{AbramskyCoecke2004}.

\section{Monads}
\label{sec:monads}
We will now outline the necessary background on monads required in later sections, and introduce the monads
that will be of particular interest. The material in this section is standard,
good sources for further background are~\cite{MacLane1978, BarrWells2005, Borceux1994b, BarrWells2005, Jacobs2012b}.
In this paper we will only be interested in monads on the category~\cset of sets and total functions,
although we will state some definitions more generally where it is cleaner to do so.
\begin{definition}[Monad]
  A \define{monad} on a category~$\mathcal{C}$ is triple consisting of:
  \begin{itemize}
  \item An endofunctor $T : \mathcal{C} \rightarrow \mathcal{C}$
  \item A \define{unit} natural transformation~$\eta : 1 \Rightarrow T$
  \item A \define{multiplication} natural transformation $\mu : TT \Rightarrow T$
  \end{itemize}
  such that the following three axioms hold:
  \begin{equation*}
    \mu \circ (\eta * T) = T\qquad\qquad
    \mu \circ (T * \eta) = T\qquad\qquad
    \mu \circ (\mu * T) = \mu \circ (T * \mu)
  \end{equation*}
\end{definition}
We now introduce the monads of interest in this paper, and relate them to computational behaviour. Similarly
to~\cite{PlotkinPower2002}, we emphasize that monads are induced by algebraic operations modelling computational,
or in our case informational, behaviour.
\begin{definition}[Lift Monad]
  The lift monad~$((-)_\bot, \eta, \mu)$ is defined as follows:
  \begin{itemize}
  \item The functor component is given by the coproduct of functors~$1 + \{ \bot \} : \cset \rightarrow \cset$.
  \item The unit and multiplication are given componentwise by:
    \begin{equation*}
      \eta_X(x) = x \qquad\qquad
      \mu_X(x) =
      \begin{cases}
        x \text{ if } x \in X\\
        \bot \text{ otherwise}
      \end{cases}
    \end{equation*}
  \end{itemize}
  The lift monad is commonly used to describe computations that can diverge.
\end{definition}
\begin{definition}[Powerset Monads]
  The \define {finite powerset monad}~$\finitepowersetmonad$ has functor component the covariant finite powerset functor.
  The unit sends an element to the corresponding singleton, and the multiplication is given by taking unions.
  The \define{non-empty finite powerset monad}~$\nonemptyfinitepowersetmonad$ arises in an analogous way by restricting the sets under consideration.
  The finite powerset monad is used to model finitely bounded non-determinism, and the non-empty finite powerset monad eliminates the possibility of divergence.
\end{definition}
\begin{definition}[Finite Distribution Monads]
  \label{def:distmonad}
  The \define{finite~distribution~monad} has functor component:
  \begin{align}
    \distmonad : \cset &\rightarrow \cset \nonumber \\
    \label{func:distmonad} X &\mapsto \{ d : X \rightarrow [0,1] \mid d \text{ has finite support and } \sum_x d(x) = 1 \}  \\
    f : X \rightarrow Y &\mapsto \lambda d\,y. \sum_{x \in f^{-1} \{y\}} d(x) \nonumber
  \end{align}
  The unit and multiplication are given componentwise by:
  \begin{equation*}
    \eta_X(x) = \delta_x \qquad\qquad \mu_X(d)(x) = \sum_{e \in \support(e)} d(e)e(x)
  \end{equation*}
  where~$\delta_x$ is the Dirac delta function and~$\support(e)$ is the support of~$e$.

  The \define{finite subdistribution monad}~$(S,\eta,\mu)$ has identical structure, except that we
  weaken the condition in equation~\eqref{func:distmonad} to:
  \begin{equation*}
    \sum_x d(x) \leq 1
  \end{equation*}
  So our finite distributions are now sub-normalized rather than normalized to~$1$.
  Both the finite distribution and subdistribution monads are used to model probabilistic computations.
  Intuitively the subdistribution monad provides scope for diverging behaviour in the ``missing''
  probability mass.
\end{definition}
\begin{remark}
  We adopt a convenient notational convention from~\cite{Jacobs2011b} and write finite distributions as formal sums~$\sum_i p_i \ket{ x_i }$,
  where we abuse the physicists ket notation to indicate the sum is a formal construction. Using
  this notation, the unit of the (sub)distribution monad is the map~$x \mapsto \ket{ x }$
  and multiplication is given by expanding out sums of sums in the usual manner.
\end{remark}
Each monad can be canonically related to a certain category of algebras.
\begin{definition}[Eilenberg-Moore Algebras]
  Let $(T,\eta,\mu)$ be a monad on~$\mathcal{C}$. An \define{Eilenberg-Moore algebra}~\cite{EilenbergMoore1965} for~$T$
  consists of an object~$A$ and a morphism~$a : T A \rightarrow A$ satisfying the following axioms:
  \begin{equation*}
    a \circ \eta_A = 1 \qquad\qquad a \circ \mu_A = a \circ T a
  \end{equation*}
  A morphism of Eilenberg-Moore algebras of type $(A,a) \rightarrow (B,b)$ is a morphism in~$\mathcal{C}$ such that:
  \begin{equation*}
    h \circ a = b \circ T h
  \end{equation*}
  The category of Eilenberg-Moore algebras and their homomorphisms will be denoted \eilmo{T}.
\end{definition}
\begin{example}
  For the monads under consideration we note that:
  \begin{itemize}
  \item The Eilenberg-Moore category of the lift monad is equivalent to the category of pointed sets and functions that
    preserve the distinguished element, denoted \pset.
  \item The Eilenberg-Moore category of the finite powerset monad is equivalent to the category of join semilattices
    and homomorphisms, denoted \jslat.
  \item The Eilenberg-Moore category of the non-empty finite powerset monad is equivalent to the category of affine join semilattices
    and homomorphisms, denoted \ajslat.
  \item The Eilenberg-Moore category of the finite distribution monad is the category of convex algebras and functions
    commuting with forming convex combinations, denoted \convex. This category has received a great deal of attention
    in for example \cite{Jacobs2010, JacobsMandemakerFurber2016, Fritz2009}.
  \item The Eilenberg-Moore category of the finite subdistribution monad is the category of subconvex algebras, that is
    algebras that can form subconvex combinations of elements in a coherent manner.
    The morphisms are functions that commute with forming subconvex combinations. We denote this category \subconvex.
  \end{itemize}
\end{example}
We consider commutative monads on the category~\cset, and specialize their definition appropriately.
\begin{definition}[Commutative Monad]
  Let~$(T, \eta, \mu)$ be a~\cset monad. There are a canonical \define{strength} and \define{costrength}
  natural transformations:
  \begin{trivlist} \item
    \begin{minipage}{0.495\textwidth}
      \begin{align*}
        \strength_{X,Y} : X \times T Y &\rightarrow T(X \times Y)\\
        (x,t) &\mapsto T(\lambda y. (x, y))(t)
      \end{align*}
    \end{minipage}
    \begin{minipage}{0.495\textwidth}
      \begin{align*}
        \costrength_{X,Y} : T X \times Y &\rightarrow T(X \times Y)\\
        (t,y) &\mapsto T(\lambda x. (x, y))(t)
      \end{align*}
    \end{minipage}
  \end{trivlist}
  The monad is said to be a \define{commutative monad} \cite{Kock1972} if the following equation holds for all~$X,Y$:
  \begin{equation*}
    \mu_{X \times Y} \circ T(\costrength_{X,Y}) \circ \strength_{T X, Y} = \mu_{X \times Y} \circ T(\strength_{X,Y}) \circ \costrength_{X,T Y}
  \end{equation*}
  This composite is then called the \define{double strength}, denoted~$\doublestrength$.
\end{definition}
\begin{remark}
  Monads are intimately related to the topic of universal algebra. The Eilenberg-Moore algebras for a~\cset monad
  can be presented by operations and equations, if we permit infinitary operations. All the monads in this paper
  are in fact \define{finitary monads}, meaning they can presented by operations of finite arity.
  Let $\phi$ and $\psi$ be operations of arities~$m$ and $n$ respectively. These operations are said to commute with
  each other if the following equation holds:
  \begin{equation*}
    \psi(\phi(x_{1,1},...,x_{1,m}),...,\phi(x_{n,1},...,x_{n,m})) = \phi(\psi(x_{1,1},...,x_{n,1}),...,\psi(x_{1,m},...,x_{n,m}))
  \end{equation*}
  If we unravel the definition of commutative monad, it says that all the operations in a presentation commute with each other.
  We can also phrase this as every operation being a homomorphism.
  More detailed discussion of connections to universal algebra and presentations can be found in \cite{Manes1976}.
\end{remark}
\begin{lemma}
  Each of the lift, powerset, finite powerset, finite non-empty power, distribution and subdistribution monads
  are commutative.
\end{lemma}
\begin{remark}
  It is interesting that all the notions of partial information and ambiguity considered in this paper give rise
  to commutative monads. Possibly we could regard this as showing these informational effects are independent
  of the order in which they are built up?
\end{remark}
Clearly many monads are not commutative:
\begin{example}
  The \define{list monad} has a functor component that sends a set to finite lists of its elements. The unit
  maps an element to the corresponding single element list and the multiplication concatenates lists of lists.
  The Eilenberg-Moore algebras of this monad are arbitrary, not necessarily commutative, monoids.
  Unsurprisingly, this monad is not commutative.
\end{example}
The following proposition captures the essential properties of Eilenberg-Moore categories of commutative monads that we will need, all in one place.
The key symmetric monoidal closed structure is due to work of Jacobs~\cite{Jacobs1994}, the other properties are well known.
\begin{proposition}
  Let~$(T,\eta,\mu)$ be a commutative monad on \cset.
  The category \eilmo{T}:
  \begin{itemize}
  \item Is a symmetric monoidal closed category.
  \item Has universal bimorphisms for the monoidal tensor.
  \item Has monoidal unit given by the free algebra $(\{*\}, !)$.
  \item The tensor product $\mu_X \otimes \mu_Y$ is isomorphic to $\mu_{X \times Y}$.
  \item Is complete.
  \item Is cocomplete.
  \end{itemize}
\end{proposition}
\begin{proof}
  Completeness and cocompleteness of categories that are monadic over \cset is standard.
  The category \cset is a SMCC via its products and exponentials. \cset is complete so we can use \cite[lemma 5.3]{Jacobs1994}.
  The additional properties of the monoidal structure come from \cite[lemmas 5.1,5.2]{Jacobs1994}.
\end{proof}
\begin{remark}
  We avoid technical discussion of universal bimorphisms, details can be found in~\cite{Kock1971, Jacobs1994}. The essential
  idea is to generalize the universal property of the tensor product of vector spaces, and their relationship to bilinear functions.
  So in the set theoretic case, homomorphisms out of our tensors will bijectively correspond to functions out of the cartesian product
  that are homomorphisms in each component separately.
\end{remark}

\section{Enriched Categories}
\label{sec:enriched}
An enriched category is a category in which the homsets have additional structure that interacts well with composition.
\begin{example}
  The following are natural examples of enriched categories:
  \begin{itemize}
  \item As a trivial example, ordinary locally small categories are~\cset-enriched.
  \item A category is poset enriched if it's homsets have a poset structure and composition in monotone with respect to that structure.
    For example the category \standardrel is poset enriched.
  \item In the categorical quantum mechanics community, a category is said to have a \define{superposition rule} ~\cite{HeunenVicary2016},
    or be a \define{process theory with sums} \cite{CoeckeKissinger2016},
    if its homsets carry commutative monoid structure that is suitably compatible with composition (and possibly additional structure).
  \end{itemize}
\end{example}
We have insufficient space for a detailed outline of the parts of enriched category theory we require, we refer the reader to~\cite{Kelly2005, Borceux1994b} for background.
The informal discussion above should hopefully be sufficient to understand the discussions in later sections.

The idea of this paper is that complete and cocomplete categories with symmetric monoidal closed structure provide a very
good base of enrichment for enriched category theory.
If we select a commutative monad that captures the linguistic feature we are interested in, we can then consider categories enriched
in the corresponding algebraic structure.

The universal bimorphism property of the monoidal structure of Eilenberg-Moore categories of commutative monads allows us
provide concrete conditions for enrichment in our categories of interest:
\begin{proposition}
  \label{prop:enrichments}
  A category $\mathcal{C}$:
  \begin{itemize}
  \item Is \pset-enriched if its homsets have pointed set structures such that:
    \begin{equation}
      \label{eq:botpreserve}
      \bot \circ f = \bot \quad\text{and}\quad f \circ \bot = \bot
    \end{equation}
  \item Is \ajslat-enriched if its homsets have affine join semilattice structures such that:
    \begin{equation}
      \label{eq:joinpreserve}
      (f \vee g) \circ h = (f \circ h) \vee (g \circ h) \quad\text{and}\quad f \circ (g \vee h) = (f \circ g) \vee (f \circ h)
    \end{equation}
  \item Is \jslat-enriched if its homsets have join semilattice lattice structures such that both the equations of~\eqref{eq:botpreserve} and~\eqref{eq:joinpreserve} hold.
  \item Is \convex-enriched if its homsets have convex algebra structures such that:
    \begin{equation}
      \label{eq:convexpreserve}
      (\sum_i p_i f_i) \circ g = \sum_i p_i (f_i \circ g) \quad\text{and}\quad f \circ (\sum_i p_i g_i) = \sum_i p_i (f \circ g_i)
    \end{equation}
  \item Is \subconvex-enriched if its homsets have subconvex algebra structures such that the equations~\eqref{eq:convexpreserve} hold for all subconvex combinations.
  \end{itemize}
\end{proposition}
For a given cocomplete symmetric monoidal closed category $\mathcal{V}$, we can form the free $\mathcal{V}$-enriched category over an ordinary category.
This construction will be exploited in the later sections, details can be found in~\cite{Kelly2005, Borceux1994b}.

\section{Unquantified Mixing}
\label{sec:unquantified}
We begin by considering probably the simplest case, in which we have incomplete information. For
example, I simply don't know the meaning of the word ``logolepsy''.
In order to model this, we enrich our category in pointed sets, with the distinguished
element denoting missing information.
\begin{definition}
  \label{def:botc}
  For category~$\mathcal{C}$ we define the category~$\mathcal{C}_\bot$ as having:
  \begin{itemize}
  \item {\bf Objects}: The same objects as $\mathcal{C}$
  \item {\bf Morphisms}: We define $\mathcal{C}_\bot(A,B) = \mathcal{C}(A,B)_\bot$
  \end{itemize}
  Identities are as in $\mathcal{C}$. Composition is given as in~$\mathcal{C}$, extended
  with the rules:
  \begin{equation*}
    \bot \circ f = \bot \quad\text{and}\quad g \circ \bot = \bot
  \end{equation*}
\end{definition}
\begin{proposition}
  For a category~$\mathcal{C}$, $\mathcal{C}_\bot$ is the free pointed set enriched category over~$\mathcal{C}$.
\end{proposition}
\begin{theorem}
  If~$\mathcal{C}$ is a compact closed category then~$\mathcal{C}_\bot$ is
  a compact closed category. The monoidal structure on morphisms extends that in~$\mathcal{C}$ as in~\eqref{eq:botmonoidal}.
  \begin{trivlist}\item
    \begin{minipage}{0.495\textwidth}
      \begin{equation}
        \label{eq:botmonoidal}
        f \otimes_\bot f' =
        \begin{cases}
          \bot \text{ if } f = \bot \text{ or } f' = \bot\\
          f \otimes f' \text{ otherwise }
        \end{cases}
      \end{equation}
    \end{minipage}
    \begin{minipage}{0.495\textwidth}
      \begin{equation}
        \label{eq:botdagger}
        f^{\dagger_\bot} =
        \begin{cases}
          \bot \text{ if } f = \bot\\
          f^\dagger \text{ otherwise }
        \end{cases}
      \end{equation}
    \end{minipage}
  \end{trivlist}
  There is an identity and surjective on objects strict monoidal embedding~$\mathcal{C} \rightarrow \mathcal{C}_\bot$.

  If~$\mathcal{C}$ is a dagger compact closed category then~$\mathcal{C}_\bot$ is a dagger compact closed category with the
  dagger extending that of~$\mathcal{C}$ as in~\eqref{eq:botdagger}.  
\end{theorem}
\begin{proof}
  We sketch the basic ideas.
  We must check that the extended definitions of the monoidal product and if necessary the dagger
  are functorial. We then wish to inherit the associator, left and right unitor and symmetry from
  the base. In order to do so we must check they remain natural with respect to the extended
  functor actions on morphisms. Cups and caps and other structure and axioms are then broadly speaking inherited
  from~$\mathcal{C}$.
\end{proof}
Although~$\mathcal{C}_\bot$ gives a new compact closed category, it is not particularly exciting. As soon
as we compose or tensor with a~$\bot$ element, the whole term becomes~$\bot$. This is consistent with the
intuition that if we have no idea about part of the information we require, we cannot know the
whole either.

We took the opportunity to sketch the proof that $\mathcal{C}_\bot$ is compact closed as it is easiest to
follow in this simple case. Later proofs of similar claims are analogous.
\begin{remark}
Although we have inherited some good properties from the base category in the construction~\ref{def:botc}, clearly
not everything can be preserved. We are expanding the morphisms between each pair of objects, for example
in~$\standardrel_\bot$ we now have three scalars, which we can interpret as true, false and unknown. This expansion
will interfere with (co)limits from the base category. For example if we have a zero object in~$\mathcal{C}$,
that object will no longer be a zero object in~$\mathcal{C}_\bot$ as there will be 2 morphisms of type $0 \rightarrow 0$.
\end{remark}
Noting the similarity with the behaviour of the~$\bot$ elements with that of zero morphisms in categories with zero objects, we note that:
\begin{lemma}
  Every category with a zero object is \pset-enriched.
\end{lemma}
Although the lift monad and \pset-enrichment are extremely straightforward, we shall return to them later, in interaction with different types of ambiguity.

If the previous model captured incomplete information, we now move to consider ambiguity. In this case
we intend situations where several things are possible, for example a bat is either a winged mammal or
sporting equipment. In particular, unlike the previous model, we have complete information about the available possibilities,
we are simply unaware of which one applies.
If we don't have any sense of the relative likelihoods of the possibilities, we
are just left with a non empty finite set of alternatives, and this points us in the direction
of the monad~$\nonemptyfinitepowersetmonad$.
\begin{definition}
  For category~$\mathcal{C}$ we define the category~$\mathcal{C}_{\nonemptyfinitepowersetmonad}$ as having:
  \begin{itemize}
  \item {\bf Objects}: The same objects as~$\mathcal{C}$
  \item {\bf Morphisms}: We define~$\mathcal{C}_{\nonemptyfinitepowersetmonad}(A,B) = \nonemptyfinitepowersetmonad(\mathcal{C}(A,B))$
  \end{itemize}
  We define composition as follows:
  \begin{equation*}
    V \circ U = \{ v \circ u \mid v \in V, u \in U \}
  \end{equation*}
  Identities are then given by the singletons containing the identities from~$\mathcal{C}$.
\end{definition}
\begin{proposition}
  For a category~$\mathcal{C}$, $\mathcal{C}_{\nonemptyfinitepowersetmonad}$ is the free affine join semilattice enriched category over~$\mathcal{C}$.
\end{proposition}
\begin{theorem}
  If~$\mathcal{C}$ is a compact closed category then~$\mathcal{C}_{\nonemptyfinitepowersetmonad}$ is
  a compact closed category. The action of the tensor on morphisms extends that of~$\mathcal{C}$ as in~\eqref{eq:affinemonoidal}.
  \begin{trivlist}\item
    \begin{minipage}{0.495\textwidth}
      \begin{equation}
        \label{eq:affinemonoidal}
        U \otimes_{\nonemptyfinitepowersetmonad} U' = \{ u \otimes u' \mid u \in U, u' \in U' \}
      \end{equation}
    \end{minipage}
    \begin{minipage}{0.495\textwidth}
      \begin{equation}
        \label{eq:affinedagger}
        U^{\dagger_{\nonemptyfinitepowersetmonad}} = \{ u^\dagger \mid u \in U \}
      \end{equation}
    \end{minipage}
  \end{trivlist}
  There is an identity and surjective on objects strict monoidal embedding~$\mathcal{C} \rightarrow \mathcal{C}_{\nonemptyfinitepowersetmonad}$.

  If $\mathcal{C}$ is a dagger compact closed category then~$\mathcal{C}_{\nonemptyfinitepowersetmonad}$ is a dagger compact closed category with the
  dagger extending that of $\mathcal{C}$ as in~\eqref{eq:affinedagger}.
\end{theorem}
As a final possibility, what if we wish to consider both ambiguous and incomplete information? It
is then natural to consider the finite powerset monad as the source of our enrichment. The definition
is almost identical to that of $\mathcal{C}_{\nonemptyfinitepowersetmonad}$:
\begin{definition}
  For category~$\mathcal{C}$ we define the category~$\mathcal{C}_{\finitepowersetmonad}$ as having:
  \begin{itemize}
  \item {\bf Objects}: The same objects as~$\mathcal{C}$
  \item {\bf Morphisms}: We define~$\mathcal{C}_{\finitepowersetmonad}(A,B) = \finitepowersetmonad(\mathcal{C}(A,B))$
  \end{itemize}
  We define composition as follows:
  \begin{equation*}
    V \circ U = \{ v \circ u \mid v \in V, u \in U \}
  \end{equation*}
  Identities are then given by the singletons containing the identities from~$\mathcal{C}$.
\end{definition}
\begin{proposition}
  For a category~$\mathcal{C}$, $\mathcal{C}_{\finitepowersetmonad}$ is the free join semilattice enriched category over~$\mathcal{C}$.
\end{proposition}
\begin{theorem}
  If~$\mathcal{C}$ is a compact closed category then~$\mathcal{C}_{\finitepowersetmonad}$ is
  a compact closed category. The action of the tensor on morphisms extends that of~$\mathcal{C}$ as in~\eqref{eq:jslatmonoidal}.
  \begin{trivlist}\item
    \begin{minipage}{0.495\textwidth}
      \begin{equation}
        \label{eq:jslatmonoidal}
        U \otimes_{\finitepowersetmonad} U' = \{ u \otimes u' \mid u \in U, u' \in U' \}
      \end{equation}
    \end{minipage}
    \begin{minipage}{0.495\textwidth}
      \begin{equation}
        \label{eq:jslatdagger}
        U^{\dagger_{\finitepowersetmonad}} = \{ u^\dagger \mid u \in U \}
      \end{equation}
    \end{minipage}
  \end{trivlist}
  There is an identity and surjective on objects strict monoidal embedding~$\mathcal{C} \rightarrow \mathcal{C}_{\finitepowersetmonad}$.

  If $\mathcal{C}$ is a dagger compact closed category then~$\mathcal{C}_{\finitepowersetmonad}$ is a dagger compact closed category with the
  dagger extending that of $\mathcal{C}$ as in~\eqref{eq:jslatdagger}.
\end{theorem}
\begin{definition}
  We consider two sub-classes of monad:
  \begin{itemize}
  \item A monad is said to be \define{affine} \cite{Lindner1979, Jacobs1994} if the component of its unit at the terminal object is an isomorphism.
  \item A monad is said to be \define{relevant} \cite{Jacobs1994} if~$\doublestrength \circ \delta = T \delta$
  \end{itemize}
  Techniques for extracting the affine and relevant parts of a commutative monad can be found in~\cite{Jacobs1994}.
\end{definition}
\begin{remark}
  Again we can think about the notion of affine monad in terms of presentations, as we did with commutativity.
  An operation is said to be idempotent if:
  \begin{equation*}
    \psi(x,x,...,x) = x
  \end{equation*}
  An algebraic theory is affine if all its operations are idempotent. In particular this means the theory can
  have no constants or non-trivial unary operations. It makes intuitive sense that descriptions
  of ambiguity should lead to affine algebraic theories. We cannot just conjure up elements out of thin air,
  and being confused between~$x$ and~$x$ should provide the same information as knowing~$x$ directly.
\end{remark}
We now note a fundamental relationship between the three monads considered in this section.
\begin{remark}
As observed in~\cite{Jacobs1994}, $\nonemptyfinitepowersetmonad$ and the lift monad are the affine and relevant
parts of the finite powerset monad. In fact the finite powerset monad can be constructed from the non-empty
finite powerset monad and the lift monad using a distributive law~\cite{Beck1969}, and so in a mathematical
sense it is precisely the description of incomplete information combined with non-quantitative ambiguity.
A similar pattern will be repeated in the next section.
\end{remark}

\section{Quantified Mixing}
\label{sec:quantified}
We now move to the setting that has typically been considered in categorical models of mixing and ambiguity
until now, probabilistic mixtures. Here we return to the situation where our state of knowledge is for example
that the word ``bank'' suggests with~$90\%$ confidence a financial bank and~$10\%$ confidence a river bank.
We now have quantitative information, and it should be possible to encode this information in our homsets.
\begin{definition}
  For category~$\mathcal{C}$ we define the category~$\mathcal{C}_\distmonad$ as having:
  \begin{itemize}
  \item {\bf Objects}: The same objects as $\mathcal{C}$
  \item {\bf Morphisms}: We define~$\mathcal{C}_\distmonad(A,B) = \distmonad(\mathcal{C}(A,B))$
  \end{itemize}
  Composition is given as follows:
  \begin{equation*}
    \sum_j q_j \ket{ g_j } \circ \sum_i p_i \ket{ f_i } = \sum_{i,j} p_i q_j \ket{ g_j \circ f_i }
  \end{equation*}
\end{definition}
\begin{example}
  Describing mixing in \cpm{\standardrel}, as discussed in the introduction, was unsatisfactory as we could
  not encode quantitative data about our state of knowledge. The category~$\standardrel_\distmonad$
  encodes a convex set of weights on its morphisms. For example the scalars in~$\standardrel_\distmonad$
  correspond to the closed real interval~$[0,1]$.
\end{example}
\begin{proposition}
  For category~$\mathcal{C}$, $\mathcal{C}_\distmonad$ is the free convex algebra enriched category over~$\mathcal{C}$.
\end{proposition}
\begin{theorem}
  If~$\mathcal{C}$ is a compact closed category then~$\mathcal{C}_\distmonad$ is
  a compact closed category. The action of the monoidal structure on morphisms
  extends that of~$\mathcal{C}$ as in~\eqref{eq:distmonoidal}.
  \begin{trivlist}\item
    \begin{minipage}{0.495\textwidth}
      \begin{equation}
        \label{eq:distmonoidal}
        \sum_i p_i \ket{f_i} \otimes_\distmonad \sum_j q_j \ket{ g_j } = \sum_{i,j} p_i q_j \ket{ f_i \otimes g_j }
      \end{equation}
    \end{minipage}
    \begin{minipage}{0.495\textwidth}
      \begin{equation}
        \label{eq:distdagger}
        (\sum_i p_i \ket{ f_i })^{\dagger_{\distmonad}} = \sum_i p_i \ket{ f^\dagger }
      \end{equation}
    \end{minipage}
  \end{trivlist}
  There is an identity and surjective on objects strict monoidal embedding~$\mathcal{C} \rightarrow \mathcal{C}_\distmonad$.

  If $\mathcal{C}$ is a dagger compact closed category then~$\mathcal{C}_{\distmonad}$ is a dagger compact closed category with the
  dagger extending that of $\mathcal{C}$ as in~\eqref{eq:distdagger}.
\end{theorem}

\begin{proposition}
  For the finite subdistribution monad we have \footnote{possibly these observations are well known or folklore, but I am unaware of a suitable prior reference}:
  \begin{itemize}
  \item The finite distribution monad is the affine part of the subdistribution monad.
  \item The lift monad is the relevant part of the subdistribution monad.
  \item The finite subdistribution monad can be constructed using a distributive law combining the finite distribution monad and the lift monad.
  \end{itemize}
\end{proposition}
\begin{remark}
  As we saw for unquantified ambiguity, quantified ambiguity is affine, so forming combinations is idempotent.
  Intuitively, quantified confusion between~$x$ and itself is the same as knowing~$x$. Similarly to the unquantified
  case, we see that the subdistribution monad is exactly the result of combining quantified ambiguity and incomplete
  information.
\end{remark}

\begin{definition}
  For category~$\mathcal{C}$ we define the category~$\mathcal{C}_\subdistmonad$ as having:
  \begin{itemize}
  \item {\bf Objects}: The same objects as $\mathcal{C}$
  \item {\bf Morphisms}: We define~$\mathcal{C}_\subdistmonad(A,B) = \subdistmonad(\mathcal{C}(A,B))$
  \end{itemize}
  Composition is given as follows:
  \begin{equation*}
    \sum_j q_j \ket{ g_j } \circ \sum_i p_i \ket{ f_i } = \sum_{i,j} p_i q_j \ket{ g_j \circ f_i }
  \end{equation*}
\end{definition}
\begin{proposition}
  For category $\mathcal{C}$, $\mathcal{C}_\subdistmonad$ is the free subconvex algebra enriched category over~$\mathcal{C}$.
\end{proposition}

\begin{theorem}
  \label{thm:subdistmonad}
  If~$\mathcal{C}$ is a compact closed category then~$\mathcal{C}_S$ is
  a compact closed category. The action of the monoidal structure on morphisms
  extends that of~$\mathcal{C}$ as in~\eqref{eq:subdistmonoidal}.
  \begin{trivlist}\item
    \begin{minipage}{0.495\textwidth}
      \begin{equation}
        \label{eq:subdistmonoidal}
        \sum_i p_i \ket{f_i} \otimes_\subdistmonad \sum_j q_j \ket{ g_j } = \sum_{i,j} p_i q_j \ket{ f_i \otimes g_j }
      \end{equation}
    \end{minipage}
    \begin{minipage}{0.495\textwidth}
      \begin{equation}
        \label{eq:subdistdagger}
        (\sum_i p_i \ket{ f_i })^{\dagger_{\subdistmonad}} = \sum_i p_i \ket{ f^\dagger }
      \end{equation}
    \end{minipage}
  \end{trivlist}
  There is an identity and surjective on objects strict monoidal embedding~$\mathcal{C} \rightarrow \mathcal{C}_\distmonad$.

  If~$\mathcal{C}$ is a dagger compact closed category then~$\mathcal{C}_{\subdistmonad}$ is a dagger compact closed category with the
  dagger extending that of~$\mathcal{C}$ as in~\eqref{eq:subdistdagger}.
\end{theorem}
The various free models of ambiguity and incomplete information that we have constructed
can be embedded into each other as follows:
\begin{proposition}
  \label{prop:embeddings}
  For category $\mathcal{C}$ there are identity and surjective on objects embeddings:
  \begin{equation*}
    \label{eq:embeddings}
    \begin{tikzpicture}[scale=0.5, node distance=2cm, ->]
      \node (tl) {$\mathcal{C}_{\nonemptyfinitepowersetmonad}$};
      \node[below of=tl, node distance=1cm] (bl) {$\mathcal{C}_{\distmonad}$};
      \node[right of=tl] (tm) {$\mathcal{C}_{\finitepowersetmonad}$};
      \node[below of=tm, node distance=1cm] (bm) {$\mathcal{C}_{\subdistmonad}$};
      \coordinate (mid) at ($(tm)!0.5!(bm)$);
      \node[right of=mid] (tr) {$\mathcal{C}_{\bot}$};
      \draw (tl) to node[above]{$E_{\nonemptyfinitepowersetmonad, \finitepowersetmonad}$} (tm);
      \draw (bl) to node[below]{$E_{\distmonad, \subdistmonad}$} (bm);
      \draw (tr) to node[above right]{$E_{\bot,\finitepowersetmonad}$} (tm);
      \draw (tr) to node[below right]{$E_{\bot,\subdistmonad}$} (bm);
    \end{tikzpicture}
  \end{equation*}
  where:
  \begin{trivlist}\item
    \begin{minipage}{0.495\textwidth}
      \begin{align*}
        E_{\bot,\finitepowersetmonad}(f) &=
        \begin{cases}
          \emptyset \text{ if } f = \bot\\
          \{ f \} \text{ otherwise }
        \end{cases}\\
        E_{\nonemptyfinitepowersetmonad, \finitepowersetmonad}(U) &= U
      \end{align*}
    \end{minipage}
    \begin{minipage}{0.495\textwidth}
      \begin{align*}
        E_{\bot,\subdistmonad}(f) &=
        \begin{cases}
          \sum_{\emptyset} \text{ if } f = \bot\\
          \ket{ f } \text{ otherwise }
        \end{cases}\\
        E_{\distmonad, \subdistmonad}(\sum_i p_i f_i) &= \sum_i p_i f_i
      \end{align*}
    \end{minipage}
  \end{trivlist}
  where~$\sum_\emptyset$ denotes the empty formal sum.
\end{proposition}
Proposition~\ref{prop:embeddings} shows that~$C_{\finitepowersetmonad}$ allows us to model
non-quantitative ambiguity and partial information. Analogously, the category~$\mathcal{C}_{\subdistmonad}$
allows the description of quantitative incomplete and ambiguous information.
It seems natural to try and describe non-quantitative ambiguity using convex mixtures by
restricting attention to uniform distributions.
As such, one might expect vertical arrows in diagram~\eqref{eq:embeddings}
with embeddings~$\mathcal{C}_{\nonemptyfinitepowersetmonad} \rightarrow \mathcal{C}_{\distmonad}$
and~$\mathcal{C}_{\finitepowersetmonad} \rightarrow \mathcal{C}_{\subdistmonad}$. This turns
out not to be functorial, essentially as a result of:
\begin{counter}
  Uniform distributions are not closed under composition in the free \convex or \subconvex enrichments.
  For uniform mixture~$e = 0.5 \ket{f} + 0.5 \ket{ f \circ f }$, $e \circ e$~is not in general a uniform mixture once we collect terms.
\end{counter}
It seems the best we can do is to observe that unquantified phenomena give general \convex-enriched and \subconvex-enriched categories, beyond
the free ones:
\begin{theorem}
  \label{thm:generalenrichment}
  For every category~$\mathcal{C}$, $\mathcal{C}_{\nonemptyfinitepowersetmonad}$~is \convex-enriched,
  and~$\mathcal{C}_{\finitepowersetmonad}$ is \subconvex-enriched. In both cases sums are given by supports:
  \begin{equation*}
    \sum_{i \in I} p_i \ket{ f_i } = \{ f_i \mid i \in I \}
  \end{equation*}
\end{theorem}
\begin{proof}
  We must confirm that taking supports gives \convex-algebra or \subconvex-algebra structure as required.
  It then remains to confirm the equations specified in proposition~\ref{prop:enrichments} hold. Both
  are straightforward calculations.
\end{proof}
Using the embeddings of proposition~\ref{prop:embeddings}
and the general enrichments noted in theorem~\ref{thm:generalenrichment}, all the phenomena
of interest are captured in some way by \subconvex-enrichment.

\section{Conclusion}
In this paper we have considered freely extending the dagger compact closed categories
used in categorical distributional models of meaning with sufficient algebraic structure
to describe incomplete and ambiguous information. This was done in a systematic manner,
constructing suitable bases for enrichment using monad theoretic principles. Our free
models effectively record and combine the details of the information we lack. The data
explicitly carried by this effectively syntactic construction suggests some algorithmic possibilities
to be explored in later work.

Clearly, models other than the free models are of interest. For example, the category
of Hilbert spaces and completely positive maps is subconvex algebra enriched, and
\standardrel is join semilattice enriched. The category~$\standardrel_{\subconvex}$ provides
a category that allows non trivial mixing of relations with scalars in~$[0,1]$.
The author is unaware of other (non-free) models involving relations
that allow non-trivial mixing with real scalars, and this remains an open question.

\subsection*{Acknowledgements}
I would like to thank Martha Lewis, Bob Coecke, Samson Abramsky and Robin Piedeleu
for feedback and discussions.
This work was partially funded by the AFSOR grant ``Algorithmic and Logical Aspects when Composing Meanings''
and the FQXi grant ``Categorical Compositional Physics''.

\bibliographystyle{eptcs}
\bibliography{mixing}

\end{document}